\documentclass[12pt]{article} 

\usepackage[margin=1in]{geometry} 
\usepackage{fontawesome5}
\usepackage{mathpazo} 

\usepackage{comment,xspace}
\usepackage{enumitem} 
\usepackage[dvipsnames]{xcolor} 
\definecolor{ptblue}{RGB}{15,76,129} 
\definecolor{ptemerald}{HTML}{009473} 
\definecolor{ptilluminating}{HTML}{F5DF4D} 
\definecolor{ptgray}{HTML}{939597} 
\usepackage{cprotect}
\usepackage[normalem]{ulem} 

\usepackage{booktabs,multirow,subcaption}

\usepackage{amsmath,amsfonts,amssymb,amsthm,mathtools,thmtools}

\DeclareMathOperator*{\argmax}{arg\,max}
\DeclareMathOperator*{\argmin}{arg\,min}
\usepackage[ruled,linesnumbered,vlined]{algorithm2e}
\SetKwRepeat{Do}{do}{while} 

\SetCommentSty{mycommentfont}

\usepackage{doi}
\usepackage[square,numbers]{natbib} 

\usepackage{hyperref}
\hypersetup{
linktocpage,
colorlinks=true,
citecolor=black, 
urlcolor=ptblue, 
linkcolor=Plum, 
}
\usepackage{cleveref} 
\usepackage{crossreftools} 
\usepackage{tikz} 
\usetikzlibrary{decorations.pathreplacing,calligraphy}

\theoremstyle{plain}
\newtheorem{theorem}{Theorem}[section]
\newtheorem{corollary}[theorem]{Corollary}

\newtheorem{lemma}[theorem]{Lemma}

\theoremstyle{definition}
\newtheorem{definition}[theorem]{Definition}

\declaretheorem[style=definition,qed=$\bigtriangleup$,sibling=theorem]{example}

\newtheorem*{theorem*}{Theorem}

\theoremstyle{remark}
\newtheorem*{remark}{\upshape\bfseries Remark}

\let\oldfaCheck\faCheck
\renewcommand{\faCheck}{\oldfaCheck\xspace}
\let\oldfaPaw\faPaw
\renewcommand{\faPaw}{\oldfaPaw\xspace}

\makeatletter 
\newcommand{\EF}[1]{\if\relax\detokenize\expandafter{\@firstofone#1{}}\relax EF\xspace\else EF#1\fi}
\makeatother

\usepackage{array}
\newcommand{\PreserveBackslash}[1]{\let\temp=\\#1\let\\=\temp}
\newcolumntype{C}[1]{>{\PreserveBackslash\centering}p{#1}}

\usepackage{authblk}

\begin{document}
\title{Fair Allocation of Two Types of Chores}

\author[1]{Haris Aziz\thanks{haris.aziz@unsw.edu.au}}
\author[1]{Jeremy Lindsay\thanks{j.lindsay@student.unsw.edu.au}}
\author[1]{Angus Ritossa\thanks{a.ritossa@student.unsw.edu.au}}
\author[1]{Mashbat Suzuki\thanks{mashbat.suzuki@unsw.edu.au}}
\affil[1]{UNSW Sydney}
\date{\vspace*{-1cm}}

\maketitle
\begin{abstract}
We consider the problem of fair allocation of indivisible chores under additive valuations. 
We assume that the chores are divided into two types and under this scenario,  we present several results. 
Our first result is a new characterization of Pareto optimal allocations in our setting and a polynomial-time algorithm to compute an envy-free up to one item (EF1) and Pareto optimal allocation. 
We then turn to the question of whether we can achieve a stronger fairness concept called envy-free up any item (EFX). We present a polynomial-time algorithm that returns an EFX allocation. Finally, we show that for our setting, it can be checked in polynomial-time whether an envy-free allocation exists or not.
\end{abstract}

\section{Introduction}

How to make allocation decisions fairly is  a fundamental question that has been examined in many fields including computer science, economics, operations research and mathematics. We consider this question in the context of allocating indivisible chores among agents where each agent has additive valuations over the chores. 

There are several formal criteria of fairness~(see e.g., \citep{BoChLa16,Moul19a}). Among the criteria, envy-freeness is referred to as the `gold-standard'~\citep{CKM+19a}. It requires that no agent prefers another agent's bundle to their own bundle. Although envy-freeness is a highly-desirable fairness concept, it poses several challenges. An envy-free allocation may not exist, and furthermore, it is NP-complete to check whether an envy-free allocation exists under additive valuations \citep{AGMW14a,BSV21a}. For this reason, a major focus on fair allocation is to find relaxations of envy-freeness. A particularly attractive relaxation of envy-freeness  is called \emph{envy-freeness up to any item (EFX)}~\citep{CKM+19a,ACIW22a}. 
However, the existence of EFX is a major open problem for goods and for chores. EFX requires that if an agent is envious of another agent, ignoring any item that lessens the envy results in the envy disappearing. A weaker concept is \emph{envy-freeness up to one item (EF1)} that requires that if an agent is envious of another agent, then there exists some item such that ignoring the item results in the envy disappearing.
It is open whether an EF1 and \textit{Pareto optimal (PO) } allocation always exists for chores. 

In view of the open problem concerning the existence of EFX as well as EF1+PO allocations and the absence of positive algorithmic results regarding envy-free allocations, we turn our attention to a natural scenario of chore allocation in which there are at most two types of chores.  We assume that the items can be divided into two groups $A$ and $B$. Chores within the same group are identical and hence a given agent has the same value for the identical items. A natural motivating example could be a group of 4 housemates allocating  monthly household chores consisting of 18 room cleaning chores and 15 cooking chores.

There are several reasons for considering the case of two chore types. Firstly, it is natural to consider restrictions on the general chore allocation under which we can achieve positive algorithmic results. For example, there are many papers that assume that agents have binary valuations for items~(see, e.g., \citep{BKV18a,DaSc15a,AzRe20a}): 0 or 1 in the case of goods and 0 and -1 in the case of chores.\footnote{Our assumption of two chore types does not assume that agents have zero as one of the two valuations. Zero valuations make many problems considerably easier.}
There are also some recent papers where agents have exactly two values in the valuation functions (bi-valued utilities)~\citep{EPS22a,GMQ22a} . In contrast, we allow the set of all agents to possibly have $2n$ different values for the set of items. Finally, two chore-types is a natural subclass of \textit{personalized bi-valued instances} ~(see, e.g., \citep{EPS22a}) in which each agent subjectively divides the items into two classes and has a corresponding value for items in each of the classes.

\subsection*{Contributions}

We give a polynomial-time algorithm for computing an EF1+fPO allocation for two chore type instances (Theorem~\ref{thm:ef1+fpo}) where fPO (fractional Pareto optimal) is a property stronger than Pareto optimality and requires Pareto optimality among all fractional outcomes. Since there are very few results known on the existence of EF1+PO allocation for chores -  as the general additive valuation setting is a major open problem  - we make concrete progress towards the problem by providing an affirmative answer in a restricted case. En route to our result, we also give a novel characterization of all fPO allocations in our setting. 

We prove that for two chore type instances an EFX allocation exists and can be computed in polynomial-time (\Cref{thm: efx}). Our algorithm differs significantly from the natural adaptation of the goods algorithm of \citet{GMV22a} and other existing approaches as they fail to produce an EFX allocation in our setting. Since the existence of EFX allocations for chores is not known even in the restricted setting of three agents with additive valuations, we remark that our work contributes towards the body of literature which explores this question in restricted settings. 

We show that there exists a polynomial-time algorithm to check whether an envy-free allocation exists in the two chore types setting (\Cref{thm: finding EF in polynomial time}). Note that this problem is NP-hard for general additive instances of indivisible chores \citep{BSV21a}. Table~\ref{table:summary} summarizes existence and complexity results under additive valuations and Figure~\ref{fig:relations} summarizes the logical relations and compatibility of the key concepts that we consider.

 \begin{figure}[h!]
     \begin{center}
        \scalebox{0.85}{
\begin{tikzpicture}

  \tikzstyle{onlytext}=[]

  \node[onlytext] (fair) at (0,0) {\begin{tabular}{c}\textbf{fairness}\end{tabular}};
  \node[onlytext] (fair) at (6,0) {\begin{tabular}{c}\textbf{efficiency}\end{tabular}};

  \draw[-, line width=1pt] (-1,-1) -- (8,-1) ;

  \node[onlytext] (EFX) at (0,-2) {\begin{tabular}{c}{EFX}\end{tabular}};

  \node[onlytext] (EF1) at (0,-5) {\begin{tabular}{c}{EF1}\end{tabular}};

\draw[->, line width=1pt] (EFX) -- (EF1) ;

  \node[onlytext] (fPO) at (6,-2) {\begin{tabular}{c}{fPO}\end{tabular}};

  \node[onlytext] (PO) at (6,-5) {\begin{tabular}{c}{PO}\end{tabular}};

\draw[->, line width=1pt] (fPO) -- (PO) ;

	    \draw [line width=20pt,opacity=0.2,green,line cap=round,rounded corners] (-0.5,-5) -- (0.5,-5)  -- (5.5,-2) -- (6.5,-2) ;

	    \draw [line width=20pt,opacity=0.2,green,line cap=round,rounded corners] (-0.2,-2) -- (0.2,-2);

 \draw[-,pink, dotted, line width=1.5pt, rounded corners, line cap=round] (-1,-2.4) -- (7,-2.4) -- (7,-1.5) -- (-1,-1.5) -- (-1,-2.4);

\end{tikzpicture}
 }
\end{center}
 \caption{\label{fig:relations} Logical relations between fairness and efficiency concepts.
 An arrow from (A) to (B) denotes that (A) implies (B). For our setting of 2 chore types, the properties in a connected solid green shape can be simultaneously satisfied, and  the combined properties in connected dotted pink are impossible to simultaneously satisfy.
 }
\end{figure}
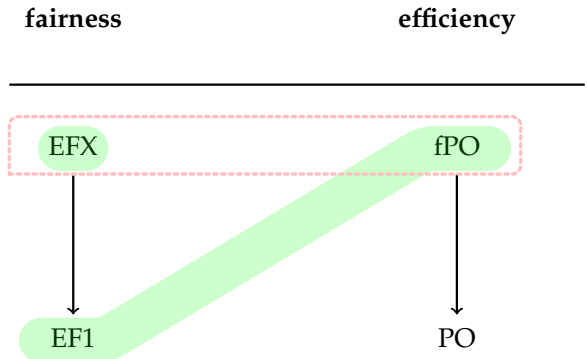

							   \begin{table*}[h!]
								   \begin{center}
								   
   \scalebox{1}{
{ \begin{tabular}{lccc}
					\toprule

				&EF1 \& PO&EFX\\
					\midrule
Chores: general&existence open&existence open\\
Chores: personalised bi-valued&existence open&existence open\\
Chores: bi-valued&in P, exists \citep{EPS22a,GMQ22a}&existence open\\
Chores: binary&in P, exists&in P, exists\\
Chores: 2 item types&in P, exists (\Cref{thm:ef1+fpo})&in P, exists (\Cref{thm: efx})\\
								   	\bottomrule
								   \end{tabular}
								   }
}
								   \end{center}
								   \caption{Existence and complexity results under additive valuations} 
								   \label{table:summary}
								   \end{table*}

\section{Related Work}

	Given that an envy-free allocation may not exist,  \citet{Budi11a} proposed a relaxation of envy-freeness called \emph{envy-free up to one item (EF1)}. An allocation satisfies EF1 if it is envy-free or any agent's envy for another agent can be removed if some item is ignored. Under additive utilities, 
EF1 can be achieved by a simple algorithm called the round-robin sequential allocation algorithm. Agents take turn in a round-robin manner and pick their most preferred unallocated item. The interest in EF1 was especially piqued when \citet{CKM+19a} proved that for positive additive utilities, a rule based on maximizing Nash social welfare finds an allocation that is both EF1 and Pareto optimal.

For negative additive valuations, the existence of an EF1 and PO allocation is a major open problem that \citet{Moul19a} highlighted in his survey (page 436). Except for a limited number of cases such as binary utilities, bi-valued utilities (\citep{EPS22a,GMQ22a}) and lexicographic valuations \citep{HSV+22a}, the guaranteed existence of EF1 and PO allocations has not been established.

In their paper \citet{CKM+19a} also presented the concept of EFX for goods which is strictly stronger than EF1. EFX requires that if an agent $i$ is envious of another agent $j$, the envy can be removed by removing any item of $j$ that is desirable to $i$.  The concepts have been adapted for the case of chores or generalized to the case of mixed goods and chores~(see e.g., \citep{ACIW22a,AMS20a}).  \citet{Proc20a} writes that the existence of EFX allocations is the biggest problem in fair division. 

There are several papers that have explored the question concerning the existence of EFX allocations and have provided partial results. 
It is well-understood that EFX allocations exists for identical valuations. 
\citet{CGM20a} proved that an EFX allocation exists for the case 3 agents and goods. 
\citet{Maha20a} showed that when items are goods and the agents have at most 2 types of valuation functions, then there exists an EFX allocation. Some of the results on sufficient conditions for the existence of EFX allocations have been extended to more general valuations~\citep{Maha21a}.
On the other hand, \citet{HSV+22a} showed that when there are mixed goods and chores, then an EFX allocation may not exist. In this paper, we focus on EFX allocation of chores and identify conditions under which an EFX allocation exists. 
\citet{ZhWu21a} presented algorithms that provide approximation of EFX for chores.
\citet{LLW21a} considered PROPX which is a weaker property than EFX in the context of chores and they proposed algorithms for PROPX allocation of  chores. One particular paper~\citep{GMV22a} focusses on positive valuations and among other results, presents an algorithm to compute an EFX allocation when there are at most two item types. The approach does not extend to the case of chores and our corresponding result requires a different approach and argument.

\citet{GMQ22a} and \citet{EPS22a} examine problems in which agents have negative bi-valued valuations\footnote{Each agent $i$ and item $o$, the valuation is either some value $a$ or $b$.}, and they both present a polynomial-time algorithm to compute an EF1 and Pareto optimal allocation. \citet{EPS22a} also showed that for a subclass of personalised bi-valued allocations an MMS fair allocation can always be computed. Previously, \citet{ABL+19a} characterized Pareto optimal allocations for positive bi-valued valuations.

\section{Preliminaries}

Let $M$ be a set of $m$ indivisible chores, and $N$ be a set of $n$ agents. Each agent $i \in N$ has a valuation function $v_i : M \rightarrow \mathbb{R}_{\leq 0}$, where $v_i(r)$ indicates $i$'s value for chore $r\in M$. Throughout the paper we assume that the valuation functions are additive, i.e., for each agent $i\in N$ and for each set of chores $S\subseteq M$,  $v_i(S)= \sum_{r\in S }v_i(r)$. Our main focus is to study the following class of instances:

\begin{definition} A fair division instance $I=(N,M,v)$ is \textit{two chore types} if the item set can be partitioned into two sets $A$ and $B$ with $M = A\cup B$, such that for each $i \in N$ we have $v_i(r)= v_i(r')$ for all $r,r'\in A$, and $v_i(h)=v_i(h')$ for all $h,h' \in B$.
\end{definition}

In plain English, an instance is two chore types if there are at most two item types such that each agent is indifferent among items of the same type. Denote $v_i^A$ as agent $i$'s value for an item of type $A$, and $v_i^B$ as value for an item of type $B$. For notational convenience, we order the agents so that $\frac{v_i^A}{v_i^B} \leq \frac{v_{i+1}^A}{v_{i+1}^B}$ for all $1 \leq i < n$, where we consider $\frac{v_{i}^A}{0}$ to be $\infty$.\footnote{We assume that no agent values both item types at 0, as otherwise we can simply allocate all the chores to that agent.}
More formally, this condition can be restated as $v_i^A v_{i+1}^B \leq v_{i+1}^A v_i^B$. 
Informally, this means that agents who prefer type $A$ items have smaller indices, and agents who prefer type $B$ items have larger indices. We divide the agents into two sets $N_A$ and $N_B$, where agents in $N_A$ prefer type $A$ items and agents in $N_B$ prefer type $B$ items.
In particular, if $v_i^A \geq v_i^B$ then $i \in N_A$, and otherwise $i \in N_B$.
We say that an agent $i \in N_A$ \emph{strongly prefers} $A$ if $2v_i^A \geq v_i^B$, and define it similarly for agents in $N_B$.

A valuation function is called \textit{bi-valued} if there exist $a,b \in \mathbb{R}$ such that $v_i(h)\in \{a,b\}$ for all $i\in N$ and $h\in M$. There have been several works which focus on bi-valued valuations \citep{GMQ22a,EPS22a}. We remark that bi-valued valuations are incomparable to two chore types valuations. Two chore type instances allow the set of agents to have $2n$ different values across agents and items whereas bi-valued instances allow for exactly two.  
A generalization of both bi-valued and two chore type instances is called \textit{personalized bi-valued}, where for each agent $i\in N$ there exist $a_i,b_i \in \mathbb{R}$ such that $ v_i(h)\in \{a_i,b_i\}$ for all $h\in M$. For personalized bi-valued instances, the existence of EF1+PO or EFX allocations are not known. \hfill \break

\noindent\textbf{Allocation:} An \textit{allocation} is a partition $X=(X_1, . . ., X_n)$ of the item set $M$, where $X_i\subseteq M$ is the bundle allocated to agent $i\in N$. An allocation is called \emph{partial} if $\bigcup_{i \in N} X_i \neq M$. We say that the allocation is \textit{fractional} if items are allocated (possibly) fractionally such that no more than one unit of each chore is allocated. In a fractional allocation, the valuation that an agent derives from an item is directly proportional to the fraction of that item that they are allocated. 
Observe that for two chore type instances any bundle can be succinctly represented by the number of items of each type in the bundle. Thus we denote $X_i = (\alpha_i,\beta_i) $ where $\alpha_i$ is the number of type $A$ items and $\beta_i$ is the number of type $B$ items in agent $i$'s bundle. We write $(\alpha, \beta) \uplus (\alpha', \beta')$ to denote the set $(\alpha+\alpha', \beta+\beta')$ for convenience.\hfill \break
   
\noindent\textbf{Fairness Notions:} An allocation $X=(X_1, . . ., X_n)$ is \emph{envy-free (EF)} if for any agents $i,j\in N$, we have $v_i(X_i)\geq   v_i(X_j)$. It is easy to see that EF allocations may not exist in general\footnote{Consider an instance where there is one chore and two agents who have negative values for the chore}. As a result weaker fairness notions EF1 and EFX have been introduced. An allocation $X$ is \textit{envy-free up to one chore }(EF1) if for any agents $i,j\in N$, where $X_i \neq \emptyset$, there exists a chore  $ h\in X_i$ such that $v_i(X_i\setminus h)\geq   v_i(X_j)$. An allocation $X$ is \textit{envy-free up to any chore }(EFX) if for any agents $i,j\in N$, and for any chore  $h\in X_i$ with $v_i(h)<0$, we have $v_i(X_i\setminus h)\geq v_i(X_j)$.
	
Observe that EFX implies EF1, but not vice versa. We say that an agent $i$ \textit{EF1-envies} (respectively \textit{EFX-envies}) another agent $j$ if $i$ envies $j$ and this envy is not EF1 (respectively EFX). \hfill \break 

\noindent\textbf{Efficiency Notions:} An allocation $Y$ \emph{Pareto dominates} another allocation $X$ if $v_i(Y_i) \geq v_i(X_i)$ for all agents $i$ and there exists an agent $j$ such that $v_j(Y_j) > v_j(X_j)$.
An allocation is \emph{Pareto optimal} (PO) if it is not Pareto dominated by any allocation.
An allocation is \emph{fractionally Pareto optimal} (fPO) if it is not Pareto dominated by any fractional allocation.
Note that an fPO allocation is also PO, but a PO allocation is not necessarily fPO. \hfill \break

In \Cref{sec:fpo} and \Cref{sec:EFX}, we assume that all agents have strictly negative valuations for both item types. We make this assumption since if there is at least one agent who values a chore at zero then both EF1+fPO and EFX allocations can be found in a straightforward way.  To see this, observe that if there is an agent $i$ with $v_i^A = 0$ and an agent $j$ with $v_j^B = 0$, then we can give all type $A$ items to agent $i$ and all type $B$ items to agent $j$. In this case, every agent values their bundle at $0$ and so this is trivially EF1+fPO and also EFX. On the other hand, without loss of generality, if there exists an agent $i$ with $v_i^A = 0$, but $v_j^B < 0$ for all agents $j$ then we assign all type $A$ items to agent $i$ and we assign the type $B$ items in a round-robin way to all the agents. This gives an EFX allocation because each agent has at most one more type $B$ item than any other agent. Additionally, this allocation is fPO since all type $A$ items were allocated to an agent who values them at zero, and so redistributing these items cannot lead to a Pareto improvement. Furthermore, if any agent were to receive fewer type $B$ items (possibly fractionally), a different agent must receive more type $B$ items, and hence no Pareto improvements are possible.

\section{ EF1+ \texorpdfstring{\lowercase{f}PO}{} }
\label{sec:fpo}

In this section, we present a polynomial-time algorithm that computes an EF1 and fPO allocation for the fair division problem with two chore type instances. En route, we give a novel characterization of fPO allocations in our setting.  \hfill \break

\noindent\textbf{Characterization of fPO Allocations}

We begin by providing a new characterization of the structure of fPO allocations by showing \Cref{lem: fPO structure}. 

\begin{lemma}
    \label{lem: fPO structure}
    Given a two chore types instance where all agents have strictly negative valuations, an allocation 
     $X = (X_1, ..., X_n)$ is fPO if and only if there exists an agent $i$ such that:
    \begin{itemize}
        \item For all agents $j$ where $\frac{v_j^A}{v_j^B} < \frac{v_i^A}{v_i^B}$, the bundle $X_j$ only contains type $A$ items.
        \item For all agents $j$ where $\frac{v_j^A}{v_j^B} > \frac{v_i^A}{v_i^B}$, the bundle $X_j$ only contains type $B$ items.
    \end{itemize}
\end{lemma}

\begin{proof}
    We first prove that any allocation which does not satisfy this criteria is not fPO. 
    In particular, consider some (potentially fractional) allocation $X$ which does not satisfy the criteria of the lemma.
    Since the criteria is not met, there must exist two agents $j$ and $k$ satisfying $\frac{v_j^A}{v_j^B} < \frac{v_k^A}{v_k^B}$, where $X_j$ has a nonzero fraction of a type $B$ item and $X_k$ has a nonzero fraction of a type $A$ item.
    Let $X_j = (\alpha_j, \beta_j)$ and $X_k = (\alpha_k, \beta_k)$.
    
    Now, consider a sufficiently small $0 < \epsilon \leq \alpha_k$ such that $\epsilon \frac{v_j^A}{v_j^B} \leq \beta_j$.
    Consider the fractional allocation $X' = (X_1', \cdots, X_n')$, where $X_j' = (\alpha_j+\epsilon, \beta_j - \epsilon \frac{v_j^A}{v_j^B})$, $X_k' = (\alpha_k-\epsilon, \beta_k + \epsilon \frac{v_j^A}{v_j^B})$ and $X_l' = X_l$ for all other agents $l$.
    Note that $v_j(X_j') = v_j(X_j) + \epsilon v_j^A - \epsilon \frac{v_j^A}{v_j^B}v_j^B = v_j(X_j)$.
    Additionally, $v_k(X_k') = v_k(X_k) - \epsilon v_k^A + \epsilon \frac{v_j^A}{v_j^B} v_k^B > v_k(X_k) - \epsilon v_k^A + \epsilon \frac{v_k^A}{v_k^B} v_k^B = v_k(X_k)$.
    Hence, the allocation $X'$ is a fractional Pareto improvement over $X$, and so $X$ is not fPO.
    
    We now prove that any allocation which satisfies the criteria of \Cref{lem: fPO structure} is fPO.
    We prove by contradiction. 
    Consider some allocation $X = (X_1, ..., X_n)$ which satisfies the criteria with some agent $i$.
    Additionally, assume that $X$ is fractionally Pareto dominated by some allocation $X' = (X_1', ..., X_n')$.
    From the previous paragraph, we can assume that $X'$ also satisfies the criteria of \Cref{lem: fPO structure} with some agent $i'$: if it did not, we could apply fractional Pareto improvements until it did.
    Let $X_j = (\alpha_j, \beta_j)$ and $X_j' = (\alpha_j', \beta_j')$ for all agents $j$.
    
    Note that for all allocations which satisfy the criteria of \Cref{lem: fPO structure}, there exists a range of possible agents $i$ for which the lemma holds.
    In particular, there are two (possibly equal) agents $i_L$ and $i_R$ such that $X$ satisfies the conditions of \Cref{lem: fPO structure} for all $i \in [i_L, i_R]$, and does not satisfy the conditions for all $i \not\in [i_L, i_R]$.
    Similarly, there exists such agents $i_L'$ and $i_R'$ for $X'$.
    We consider two cases:
    
    First, assume there exists some agent $i \in [i_L, i_R] \cap [i_L', i_R']$. 
    Then, let $N_1$ be the set of agents $j$ with $\frac{v_j^A}{v_j^B} < \frac{v_i^A}{v_i^B}$, $N_2$ be the agents $j$ with $\frac{v_j^A}{v_j^B} = \frac{v_i^A}{v_i^B}$ and $N_3$ be the agents $j$ with $\frac{v_j^A}{v_j^B} > \frac{v_i^A}{v_i^B}$.
    Then, agents in $N_1$ receive only type $A$ items in both $X$ and $X'$, and agents in $N_3$ receive only type $B$ items in both $X$ and $X'$.
    Let $X_{N_1} = \biguplus_{j \in N_1}X_j$, and define $X_{N_2}$, $X_{N_3}$, $X'_{N_1}$, $X'_{N_2}$ and $X'_{N_3}$ similarly.
    Since $X'$ Pareto dominates $X$, it follows that $|X'_{N_1}| \leq |X_{N_1}|$ and  $|X'_{N_3}| \leq |X_{N_3}|$.
    However, since $X_{N_1} \uplus X_{N_2} \uplus X_{N_3} = X'_{N_1} \uplus X'_{N_2} \uplus X'_{N_3}$, we know that $X_{N_2} \subseteq X'_{N_2}$.
    These constraints can only be satisfied if $|X'_{N_1}| = |X_{N_1}|$, $|X'_{N_2}| = |X_{N_2}|$ and $|X'_{N_3}| = |X_{N_3}|$.
    Therefore $X'$ cannot Pareto dominate $X$: at best, all agents receive the same valuation in both allocations, which is a contradiction.
    
    Otherwise, assume that $[i_L, i_R] \cap [i_L', i_R'] = \emptyset$.
    Without loss of generality, assume that $i_R' < i_L$.
    Note there must exist an agent in $[i_L, i_R]$ who received a type $A$ item in $X$: otherwise, $X$ would satisfy the conditions of $\Cref{lem: fPO structure}$ for $i = i_L-1$.
    Hence, it follows that, in $X$, not all of the type $A$ items are allocated to agents in the range $[1, i_R']$.
    However, in $X'$, all the type $A$ items are allocated to agents in the range $[1, i_R']$. 
    Therefore there must exist an agent $j \in [1, i_R']$ who receives a worse bundle in $X'$ than they do in $X$, which is a contradiction.
\end{proof}

We remark that Lemma~\ref{lem: fPO structure} allows us to restrict our attention to allocations that obey the structure outlined in the lemma. In \Cref{fig: fpo}, we give a visualisation of this structure.

\begin{figure}[h]
\centering
\begin{tikzpicture}[thick]

\draw (0,0.5) node {$\frac{v_1^A}{v_1^B} \leq ... \leq \frac{v_{i-1}^A}{v_{i-1}^B} < \frac{v_{i}^A}{v_{i}^B} = ... = \frac{v_{j}^A}{v_{j}^B} < \frac{v_{j+1}^A}{v_{j+1}^B} \leq ... \leq \frac{v_{n}^A}{v_{n}^B}$};

\draw [decorate,
    decoration = {calligraphic brace,mirror}] (-4.4,0) -- node[below] {Only type $A$} (-1.8,0);
    
\draw [decorate,
    decoration = {calligraphic brace,mirror}] (-1.2,0) -- node[below] {No restrictions} (1.2,0);
    
\draw [decorate,
    decoration = {calligraphic brace,mirror}] (1.8,0) -- node[below] {Only type $B$} (4.4,0);
    
\end{tikzpicture}

\caption{The general form of allocations which satisfy \Cref{lem: fPO structure}.}
\label{fig: fpo}

\end{figure}

\subsection*{Algorithm for EF1+fPO}

To find an EF1 and fPO allocation, it is sufficient to consider only a subset of the allocations that satisfy \Cref{lem: fPO structure}. In particular, we consider a set of allocations with the following structure.
\begin{definition}    
    An allocation $X = (X_1, ..., X_n)$ is \textit{ordered with respect to agent i} (or \textit{ordered} for short) if there exists some agent $i$ where:
    \begin{itemize}
        \item For all agents $j$ where $j < i$, the bundle $X_j$ only contains type $A$ items.
        \item For all agents $j$ where $j > i$, the bundle $X_j$ only contains type $B$ items.
    \end{itemize}
\end{definition}

We remark that all ordered allocations satisfy \Cref{lem: fPO structure}, but the converse does not necessarily hold (in particular, it does not always hold when there are multiple agents with identical preferences).

First, we consider an even more restricted class of allocations, namely \textit{split-round-robin}.
\begin{definition}
    Let $i$ be an agent such that $1 \leq i < n$. The allocation \textit{split-round-robin(i)} is the allocation formed by distributing the type $A$ items to agents 1 through $i$ in a round-robin way, and distributing the type $B$ items to agents $i+1$ through $n$ in a round-robin way. In both cases, we allocate to agents with smaller indices first.
\end{definition}

By \Cref{lem: fPO structure},  the allocation \textit{split-round-robin(i)} is fPO for all $i$. 
We introduce terminology to describe whether a \textit{split-round-robin} allocation is EF1. Let $i$ be an agent such that $1 \leq i < n$. We say that the allocation \textit{split-round-robin(i)} has \textit{$A$-envy} if there is an agent $j \leq i$ who has EF1-envy towards another agent $k > i$. Similarly, we say that the allocation \textit{split-round-robin(i)} has \textit{$B$-envy} if there is an agent $j > i$ who has EF1-envy towards another agent $k \leq i$.

Observe that \textit{split-round-robin(i)} is EF1 if and only if it does not have $A$-envy nor $B$-envy. We can now begin describing our algorithm for finding an EF1 and fPO allocation. \Cref{alg:ef1+fpo} begins by checking whether \textit{split-round-robin(i)} is EF1 for any $1 \leq i < n$. If so, then the algorithm has found an EF1 and fPO allocation. Otherwise, we create an allocation which is ordered with respect to a carefully chosen agent, who we call a \textit{split-agent}.
\begin{definition}
    An agent $i$ is a \textit{split-agent} if both of the following conditions hold:
    \begin{itemize}
        \item Either $i = 1$ or \textit{split-round-robin(i-1)} has $A$-envy, and
        \item Either $i = n$ or \textit{split-round-robin(i)} has $B$-envy.
    \end{itemize}
\end{definition}

\begin{lemma}
    \label{lem: split agent existence}
    If \textit{split-round-robin(i)} is not EF1 for all $1 \leq i < n$, then there exists a split-agent.
\end{lemma}
\begin{proof}
    Observe that if \textit{split-round-robin(i)} is not EF1 (for any $1 \leq i < n$), it must have $A$-envy or $B$-envy.    
    If neither $1$ nor $n$ are split-agents, then \textit{split-round-robin(1)} has $A$-envy and \textit{split-round-robin(n-1)} has $B$-envy.
    Hence, there must exist some $1 < i < n$ such that \textit{split-round-robin(i-1)} has $A$-envy and \textit{split-round-robin(i)} has $B$-envy.
\end{proof}

We select a split-agent $i^*$, and will create an instance that is ordered with respect to $i^*$. We now explore a useful property of ordered allocations.
\begin{lemma}
    \label{lem: EF1 for i implies EF1}
    Let $I=(N, M, v)$ be a two chore types instance and $X$ be an allocation that is ordered with respect to agent $i^*$. Consider a modified valuation profile $ \tilde{v}$, where $\tilde{v}_j = v_{i^*}$ for all $j\in N$. If $X$ is EF1 with respect to the modified valuation profile $\tilde{v}$ then it is  EF1 in the original valuation profile $v$. 
\end{lemma}
\begin{proof}  
    As $X$ is ordered with respect to agent ${i^*}$, any agent $j < {i^*}$  has only type $A$ items i.e., $X_j = (\alpha_j, 0)$. Consider now some other agent $k\in N$. We show that if agent $j$ does not EF1-envy $k$ under a modified valuation  $\tilde{v}_j = v_{i^*}$, then $j$ does not EF1-envy $k$ in the original instance.

Observe that if $\alpha_j = 0$, then agent $j$ is not allocated any chores, and thus she does not have envy towards any other agent. Hence we assume that $\alpha_j > 0$. Since $X$ is EF1 under the modified valuation profile, agent $j$ does not EF1-envy $k$ when $\tilde{v}_j = v_{i^*}$. It follows that, 
    \begin{align}\label{EF1 for i implies EF1: equation}
      \tilde{v}_j(\alpha_j - 1, 0) &= (\alpha_j - 1)v_{i^*}^A \nonumber \\ 
    &\geq \alpha_{k} v_{i^*}^A + \beta_{k} v_{i^*}^B 
    \end{align}
    Recalling  $j< {i^*}$, we have $\frac{v_j^A}{v_j^B} \leq \frac{v_{i^*}^A}{v_{i^*}^B}$. Rearranging we have that $\frac{v_j^A}{v_{i^*}^A} \leq \frac{v_j^B}{v_{i^*}^B}$. As both sides of \Cref{EF1 for i implies EF1: equation} are non-positive, it follows that $(\alpha_j-1)v_j^A \geq \alpha_{k}v_j^A + \beta_{k}v_j^B$, and hence $j$ does not EF1-envy $k$ under the original valuation function.

    We can apply a similar argument for agents $j > {i^*}$.
\end{proof}

\begin{algorithm}
\caption{Computing an EF1 and fPO allocation}
\label{alg:ef1+fpo}

\SetKwInOut{Input}{Input}
\SetKwInOut{Output}{Output}
\SetKwComment{Comment}{$\triangleright$\ }{}
        
\Input{A fair allocation instance with two chore types, where all agents have strictly negative valuations}

\Output{An allocation which is EF1 and fPO}

\For{$i \gets 1$ \KwTo $n-1$} {
    \If{\textit{split-round-robin(i)} is EF1} {
        \Return \textit{split-round-robin(i)}
    }
}

$i^* \gets$ a \textit{split agent} \Comment*{Note that such an agent is guaranteed to exist by \Cref{lem: split agent existence}.}

$X = (X_1, ..., X_n) \gets$ an allocation where $X_{i^*} = M$ and $X_j = \emptyset$ for all $j \neq {i^*}$.

\While{$X$ is not EF1} { \label{ef1+po alg while-loop}
    $j \gets$ an agent in $\argmax_{j \in N \setminus {i^*}} v_{i^*}(X_j)$
    
    \If{$j < {i^*}$} {
        Transfer a type $A$ item from $X_{i^*}$ to $X_j$ \label{ef1+fpo alg transfer line A}
    } \Else {
        Transfer a type $B$ item from $X_{i^*}$ to $X_j$ \label{ef1+fpo alg transfer line B}
    }
}
\Return $X$
\end{algorithm}

\begin{theorem}
    \label{thm:ef1+fpo}
     Given a two chore types instance, \Cref{alg:ef1+fpo} finds an allocation that is EF1 and fPO in polynomial-time.
\end{theorem}
\begin{proof}

First observe that the algorithm only outputs an ordered allocation and thus fPO by \Cref{lem: fPO structure}. Furthermore  if \textit{split-round-robin(i)} is EF1 for some $i$ then the algorithm returns an allocation that is both EF1 and fPO immediately. Thus the main challenge is to analyse the algorithm on instances where \textit{split-round-robin(i)} is not EF1 for any $1\leq i < n$. In the remainder of the proof we restrict our attention to these instances.

Recall that by \Cref{lem: split agent existence} there exists a split agent $i^*$.
At a high level the algorithm transfers items from the split agent to other agents until the allocation becomes EF1 whilst maintaining that the allocation is ordered with respect to $i^*$.

Consider now a modified valuation $\tilde{v}_j = v_{i^*}$ for all $j\in N$. We show that the algorithm outputs an EF1 allocation with respect to the modified instance. By  \Cref{lem: EF1 for i implies EF1}, the same allocation is also EF1 with respect to the original instance. In the modified instance,  there is no EF1-envy among all agents other than $i^*$ since their bundles are formed by repeatedly transferring an item to the agent with the highest valuation. In particular, if $X$ is not EF1, this must be due to EF1-envy that agent $i^*$ has for another agent, or EF1-envy that another agent has towards agent $i^*$. 

Let $X^L$ be the earliest allocation $X$ encountered in \Cref{alg:ef1+fpo} where $v_{i^*}(X_{i^*})\neq \min_{j\in N} v_{i^*}(X_j)$ holds (assuming that \Cref{alg:ef1+fpo} does not terminate prior to this). If \Cref{alg:ef1+fpo} terminates prior to $X^L$, for simplicity we say that every allocation in the algorithm is prior to $X^L$.
We show that for all allocations $X$ prior to (and including) $X^L$, no agent has  EF1-envy towards $i^*$. The statement holds for all allocations prior to $X^L$  from the definition of $X^L$.  We now prove that in the allocation $X^L$,  no agent has EF1-envy  towards $i^*$. Let $X'$ be the allocation immediately prior to $X^L$. Note that $X'$ is not EF1, or otherwise \Cref{alg:ef1+fpo} would have terminated. By definition of $X'$, we have $v_{i^*}(X'_{i^*}) = \min_{j \in N}v_{i^*}(X'_j)$. Since $X'$ is not EF1 it follows that agent ${i^*}$ must EF1-envy agent $j$, where $j \in \argmax_{j \in N \setminus \{{i^*}\}} v_{i^*}(X'_j)$. Note that $j$ is the agent who was transferred an item in \Cref{alg:ef1+fpo} when the allocation $X^L$ was created. 
Since the bundle $X^L_{i^*}$ has one less item than $X'_{i^*}$ and agent ${i^*}$ had EF1-envy towards $j$ when the allocation was $X'$, it follows that $v_{i^*}(X^L_{i^*}) < v_{i^*}(X_j')$. For all agents $k\neq j,i^*$, their bundle is unchanged between $X'_k$ and  $X^L_k$. Because agent $k$ does not EF1-envy the bundle $X'_j$ they do not EF1-envy the even worse bundle $X^L_{i^*}$. Therefore in the allocation $X_L$, no agent has EF1-envy towards $i^*$. \break

\noindent\textbf{Claim 1:} For all allocations $X$ prior to and including $X^L$, we have that $X_{i^*}$ contains at least one item of each type. 
\begin{proof}[Proof of Claim 1] Recall that no agent has EF1 envy towards $i^*$ and thus every agent other than $i^*$ has no EF1-envy towards any agent. 

We first prove that $X_{i^*}$ has at least one type $A$ item. If ${i^*} = 1$, then this is immediately true. Otherwise, assume ${i^*} > 1$. We proceed by contradiction.
    Assume that $X_{i^*}$ has no type $A$ items. 
    Then, agents $1$ through $i^*-1$ have all the type $A$ items, and agents $i$ through $n$ have all the type $B$ items, just as in the allocation \textit{split-round-robin}$(i^*-1)$.
    However, because $i^*$ is a split-agent, we know that \textit{split-round-robin}$(i^*-1)$ has $A$-envy. Thus there must exist some agent $j < i^*$ who has EF1-envy towards another agent $k \geq i^*$ in $X$ which is a contradiction.

We now prove that there is at least one type $B$ item. If $i^*=n$, it follows immediately. Otherwise, if $i^*<n$, we can use a symmetrical argument to the type $A$ item case. 
\end{proof}

In the next paragraph, we will show that the algorithm terminates (i.e. returns an EF1 allocation) prior to or at allocation $X^L$. Therefore, by Claim 1, whenever \Cref{ef1+fpo alg transfer line A} is reached, $X_{i^*}$ has at least one type $A$ item, and whenever \Cref{ef1+fpo alg transfer line B} is reached, $X_{i^*}$ has at least one type $B$ item. 

If the algorithm terminates prior to $X^L$ then we are done. Otherwise, if every allocation prior to $X^L$ is not EF1 then we show that $X^L$ must be EF1. By the definition of $X^L$, there exists some agent $k$ such that $v_{i^*}(X^L_k) < v_{i^*}(X^L_{i^*})$. Since no agent in $N \setminus {i^*}$ has any EF1-envy towards any other agent in $N$, it follows that $k$ does not EF1-envy any agent i.e., there exists some chore $r \in X^L_k$ such that $v_{i^*}(X^L_k\setminus r) \geq v_{i^*}(X^L_l)$ for all agents $l$. 
By Claim 1, $X^L_{i^*}$ contains at least one item of each type, and therefore contains an item $r'$ of the same type as $r$.
Therefore $v_{i^*}(X^L_{i^*} \setminus r') > v_{i^*}(X^L_k\setminus r) \geq v_{i^*}(X^L_l)$ for all $l \in N$. Hence $X^L$ is EF1 with respect to the modified instance. 

As for time complexity,  the algorithm runs in polynomial-time since the while loop on \Cref{ef1+po alg while-loop} can only run at most $m$ times. 
\end{proof}

\noindent\textbf{EFX and fPO are not always compatible}

A natural extension of \Cref{thm:ef1+fpo} is to ask whether an allocation always exists that is EFX and fPO. 
Here, we disprove this by providing an instance with no allocation that is both EFX and fPO. 

Consider an instance with 3 agents, where $v_1^A = -10$, $v_2^A = -11$, $v_3^A = -12$, and $v_1^B = v_2^B = v_3^B = -1$.
There are 3 type $A$ items and 2 type $B$ items.
For the allocation to be EFX, each agent must receive one type $A$ item.
Otherwise, one agent would receive at least 2 type $A$ items and another agent would receive no type $A$ items, which cannot be EFX.
However, if the allocation is fPO it must satisfy \Cref{lem: fPO structure} and so agent 3 must receive both type $B$ items.
However, this is not EFX.
Hence, in this instance, there does not exist any allocation that is both EFX and fPO.

Due to this nonexistence result, we instead consider the question of whether an EFX allocation always exists.

\section{EFX}
\label{sec:EFX}

In this section, we give an algorithm to compute an EFX allocation of chores when there are two item types. Our first observation is that important algorithms for chore allocation as well natural adaptations for fair allocation of goods to the case of chores do not give EFX guarantees even for two item types. These include two algorithms (``The Top-trading Envy Cycle Elimination Algorithm'' and ``The Bid-and-Take Algorithm'') for PROPX allocations by  \citet{LLW21a} as well as an adaptation the algorithm of \citet{GMV22a} to the case of chores.
This is detailed in \Cref{section: Failure of Existing EFX Approaches}.

The main result of this section is \Cref{thm: efx}, which we use the remainder of this section to prove.
\begin{theorem}
    \label{thm: efx}
   For two chore type instances, an EFX allocation always exists and can be found in polynomial-time.
\end{theorem}

\subsection{Allocation algorithm when \texorpdfstring{ $|A| \leq |N_A|$ or $|B| \leq |N_B|$}{lg}}
\label{subsection: |A| <= |N_A| or |B| <= |N_B|}

The algorithm in \Cref{subsection: |A| > |N_A| and |B| > |N_B|} requires $|A| > |N_A|$ and $|B| > |N_B|$, and so we begin with an algorithm for when this does not hold. Assume without loss of generality that $|A| \leq |N_A|$. Let $k = \left\lfloor \frac{|B|}{n} \right\rfloor$.
We allocate $k$ type $B$ items to all agents, and let $b$ be the number of unallocated type $B$ items.
Note that $0 \leq b < n$.
We consider two cases, depending on $b$.

\paragraph{Case 1: $b \leq |N_B|$}
We allocate 1 more type $B$ item to any $b$ agents from $N_B$ and allocate up to 1 type $A$ item to agents in $N_A$. Let $N_B' \subseteq N_B$ be the set of agents who receive an extra type $B$ item and let $N_A' \subseteq N_A$ be the set of agents who receive a type $A$ item. Then, the allocation is:
\begin{equation*}
    X^*_i =
    \begin{cases}
      (1, k) & \text{for $i \in N_A'$,} \\
      (0, k) & \text{for $i \in N_A \setminus N_A'$,} \\
      (0, k) & \text{for $i \in N_B \setminus N_B'$,} \\
      (0, k+1) & \text{for $i \in N_B'$.} \\
    \end{cases}
\end{equation*}
Note that this allocation is EFX, completing this case.

\paragraph{Case 2: $b > |N_B|$}
Let $N_A' = \{ \,  i \in N_A : i > n-b \,  \}$ and note that $|N_A'| = b - |N_B|$.
We allocate one more type $B$ item to all agents in $N_B \cup N_A'$.
This gives us the following partial allocation, where all type $B$ items are allocated:
\begin{equation*}
    X^*_i =
    \begin{cases}
      (0, k) & \text{for $i \in N_A \setminus N_A'$,} \\
      (0, k+1) & \text{for $i \in N_A'$,} \\
      (0, k+1) & \text{for $i \in N_B$.} \\
    \end{cases}
\end{equation*}
Now, let $l \geq 1$ be the largest integer such that $l v_i^A \geq v_i^B$ for all $i \in N_A \backslash N_A'$.
We assign type $A$ items to agents in $N_A \backslash N_A'$ in a round-robin way until no type $A$ items remain or all agents in $N_A \backslash N_A'$ have $l+1$ type $A$ items. 
Note that this (potentially partial) allocation is EFX due to the selection of $l$.
In particular:
\begin{itemize}
	\item All agents in $N_A' \cup N_B$ have the same bundle, and so there is no envy between these agents. Additionally, any envy from an agent $i \in N_A' \cup N_B$ towards an agent $i' \in N_A \setminus N_A'$ disappears if one type $B$ item is removed from the bundle of agent $i$. Hence there is no EFX-envy from agents in $N_A' \cup N_B$ towards any other agent.
	\item The bundles of agents in $N_A \setminus N_A'$ differ by at most one type $A$ item, and so there is no EFX-envy between any of these agents. Additionally, since $l v_i^A \geq v_i^B$ for all $i \in N_A \backslash N_A'$, any envy between an agent $i \in N_A \backslash N_A'$ towards an agent $j \in N_A' \cup N_B$ disappears if one type $A$ item is removed from the bundle of agent $i$. Hence there is no EFX-envy from agents in $N_A \backslash N_A'$ towards any other agent.
\end{itemize}

Let $a \leq |N_A'|$ be the number of unallocated type $A$ items. 
If $a = 0$, we have an EFX allocation.
Otherwise, assume that $a > 0$. Then, the current EFX partial allocation is:
\begin{equation*}
    X^*_i =
    \begin{cases}
      (l+1, k) & \text{for $i \in N_A \setminus N_A'$,} \\
      (0, k+1) & \text{for $i \in N_A'$,} \\
      (0, k+1) & \text{for $i \in N_B$.} \\
    \end{cases}
\end{equation*}

By the definition of $l$, we know there exists an agent $i \in N_A \setminus N_A'$ such that  $(l+1)v_i^A < v_i^B$.
Since the agents in $N_A \setminus N_A'$ are the agents with the smallest $\frac{v_i^A}{v_i^B}$, it must hold that $(l+1)v_i^A < v_i^B$ for all $i \in N_A' \cup N_B$.
Therefore, no agent $i \in N_A'$ envies any agent $i' \in N_A \setminus N_A'$ and so it follows that for all $i \in N_A'$, agent $i$ is envy-free.
We can then complete the allocation by selecting $a$ agents arbitrarily from $N_A'$ and allocating each of them one type $A$ item.

\subsection{Allocation algorithm when \texorpdfstring{ $|A| > |N_A|$ and $|B| > |N_B|$ }{ } }
\label{subsection: |A| > |N_A| and |B| > |N_B|}

In this section, we prove that \Cref{alg:efx} always finds an EFX allocation in polynomial-time.
We assume without loss of generality that $|N_A| \geq |N_B|$.

We begin with an overview of \Cref{alg:efx}. \Cref{alg:efx} starts by computing an EFX partial allocation $X^*$ on \Cref{efx alg initial allocation}.
In this initial allocation, all type $B$ (and potentially some type $A$) items are allocated.
\Cref{alg:efx} then applies one of following two update rules until all type $A$ items are allocated:
\begin{itemize}
    \item Rule 1 (\Cref{efx alg rule 1}). Let $a$ be the number of unallocated type $A$ items and let $X' = (X_1', ..., X_n')$ be an allocation where $X_i' = X_i$ for all $i \in N_A$ and $X_j' = X_j \uplus (1, 0)$ for all $j \in N_B$. 
    If $a \geq |N_B|$ and $X'$ is EFX, then set $X$ to be $X'$. We refer to the condition ``$X'$ is EFX'' as the ``EFX condition of Rule 1''.
    
    \item Rule 2 (\Cref{efx alg rule 2}).
    If Rule 1 does not apply, then let $i \in N_A$ be an agent who is envy-free (we will prove that such an agent always exists under our choice of $X^*$).
    We allocate a type $A$ item to $i$.
\end{itemize}

Note that both rules preserve EFX. 
In particular, Rule 1 preserves EFX by definition, and Rule 2 preserves EFX because any envy that agent $i$ has will disappear if a single type $A$ item is removed from their bundle.
Hence, if \Cref{alg:efx} returns, then the returned allocation will be EFX.
Additionally, \Cref{alg:efx} runs in polynomial-time because the update rules will be applied at most $m$ times.

However, it is not guaranteed that the updates rules can always be applied for every choice of $X^*$: \Cref{ex:EFX-alg-fails} demonstrates a case where neither rule can be applied.
Therefore, the initial allocation $X^*$ must be chosen carefully so that a situation similar to \Cref{ex:EFX-alg-fails} never occurs. 
In particular, for the chosen initial allocation $X^*$ we must show that whenever \Cref{efx alg rule 2 agent} is reached, there always exists an agent $i \in N_A$ where $v_i(X_i) \geq v_i(X_j)$ for all $j \in N$.
We introduce some terminology to reason about this: if there exists such an agent $i$, we say that ``Rule 2 can be applied''.
If it is possible to apply Rule 2 $k$ times consecutively, then we say that ``Rule 2 can be applied $k$ times''.
Note that we use these terms regardless of whether Rule 1 can be applied. 

\begin{example}
\begin{table}[h]
    \centering
   \begin{tabular}{C{0.09\textwidth}C{0.09\textwidth}C{0.09\textwidth}C{0.09\textwidth}}
    	\toprule
	\textbf{Agents} & \textbf{$v_i^A$} & \textbf{$v_i^B$} & \textbf{$X_i$}\\
	\midrule
	1 & $-1$ & $-5$ & $(1, 1)$ \\
 	2 & $-5$ & $-1$ & $(0, 1)$ \\
	3 & $-5$ & $-1$ & $(0, 2)$ \\
	\bottomrule
    \end{tabular}
	
\end{table}
\label{ex:EFX-alg-fails}
An instance with an EFX allocation $X$. If there is $a = 1$ unallocated type $A$ item, then neither update rule can be applied.
In particular, Rule 1 cannot be applied because there are insufficient unallocated items. Rule 2 cannot be applied because agent 1 would EFX-envy agent 2 if the rule were to be applied.
\end{example}

\begin{algorithm}
\caption{Computing an EFX allocation}
\label{alg:efx}
\DontPrintSemicolon

\SetKwInOut{Input}{Input}
\SetKwInOut{Output}{Output}
\SetKwComment{Comment}{$\triangleright$\ }{}
        
\Input{A fair allocation instance with two chore types, where all agents have strictly negative valuations and $|N_A| \geq |N_B|$}

\Output{An EFX allocation}

\If{$|A| \leq |N_A|$ or $|B| \leq |N_B|$} {
    \Return the allocation described in \Cref{subsection: |A| <= |N_A| or |B| <= |N_B|}
}

$X = (X_1, ..., X_n) \gets X^*$, an initial partial EFX allocation, described in \Cref{subsection: deciding the initial allocation} \label{efx alg initial allocation}

\While{$X$ is a partial allocation} {
    $a \gets$ the number of unallocated type $A$ items
    
    $X' = (X_1', ..., X_n') \gets$ an allocation where $X_i' = X_i$ for all $i \in N_A$ and $X_j' = X_j \uplus (1, 0)$ for all $j \in N_B$
    
    \If{$a \geq |N_B|$ and $X'$ is EFX} {
        $X \gets X'$ \label{efx alg rule 1} \Comment*{Rule 1}
    } \Else {
        $i \gets$ an agent in $N_A$ where $v_i(X_i) \geq v_i(X_j)$ for all $j \in N$ \label{efx alg rule 2 agent}
        
        $X_i \gets X_i \uplus (1, 0)$ \label{efx alg rule 2} \Comment*{Rule 2}
    }
}

\Return $X$
\end{algorithm}

The remainder of this section is structured as follows:
We begin by introducing some results in \Cref{no mutual envy}-\ref{lem:rule1 not applied twice} that are helpful later in the section.
We then provide conditions for $X^*$ under which \Cref{alg:efx} always finishes and returns an allocation. 
In particular, both \Cref{using the update rules} and \Cref{N_B same size} give sufficient conditions for $X^*$. 
Finally, in \Cref{subsection: deciding the initial allocation}, we show how to compute the initial allocation $X^*$. 
To do this, we must consider several cases that together cover every possible input instance for \Cref{alg:efx}.
In every case, we show that we can find an initial allocation $X^*$ that satisfies the criteria of \Cref{using the update rules} or \Cref{N_B same size}.

\begin{lemma}
    \label{no mutual envy}
    Let $i$ and $j$ be two agents, and let $X$ be an allocation. 
    If $i > j$ and $X_i$ has at least as many type $B$ items as $X_j$, then $i$ and $j$ cannot both envy each other.
    That is, if $i$ envies $j$, then $j$ does not envy $i$.
\end{lemma}
\begin{proof}
    First, note that $\frac{v_i^A}{v_i^B} \geq \frac{v_j^A}{v_j^B}$.
    Let $X_i = (\alpha_i, \beta_i)$ and $X_j = (\alpha_j, \beta_j)$.
    Assume that $i$ envies $j$, and so 
    $$\alpha_i v_i^A + \beta_i v_i^B < \alpha_j v_i^A + \beta_j v_i^B.$$
    Rearranging this gives
    $$v_i^B(\beta_i - \beta_j) < v_i^A(\alpha_j - \alpha_i).$$
    Since $\frac{v_j^B}{v_i^B} \geq \frac{v_j^A}{v_i^A}$ and the left side of the above equation is non-positive (since $\beta_i \geq \beta_j$ and $v_i^B < 0$), it follows that
    $$v_j^B(\beta_i - \beta_j) < v_j^A(\alpha_j - \alpha_i),$$
    and so 
    $$\alpha_i v_j^A + \beta_i v_j^B < \alpha_j v_j^A + \beta_j v_j^B.$$
    Hence $j$ does not envy $i$.
\end{proof}

\begin{lemma}
    \label{envy means less items}
    Let $i \in N_A$ and $j \in N_B$ be two agents, and let $X_i$ and $X_j$ be their bundles. If $X_j$ has strictly more type $B$ items than $X_i$ and $j$ EFX-envies $i$, then $|X_i| < |X_j|-1$.
\end{lemma}
\begin{proof}
    Let $X_i = (\alpha_i, \beta_i)$ and $X_j = (\alpha_j, \beta_j)$.
    Since $j$ EFX-envies $i$, 
    $$v_j((\alpha_i, \beta_i)) > v_j((\alpha_j, \beta_j-1)),$$ 
    and hence
    $$\alpha_i v_j^A + \beta_i v_j^B > \alpha_j v_j^A + (\beta_j-1)v_j^B.$$
    Rearranging this gives
    $$(\alpha_i - \alpha_j)v_j^A > ((\beta_j-1) - \beta_i)v_j^B.$$
    Noting that $(\beta_j-1) - \beta_i \geq 0$ and $0 > v_j^B \geq v_j^A$, it follows that 
    $$\alpha_i - \alpha_j < (\beta_j-1) - \beta_i,$$
    and so $\alpha_i+\beta_i < \alpha_j+\beta_j-1$, implying that $|X_i| < |X_j|-1$. 
\end{proof}

\begin{lemma}
    \label{lem:less items implies no envy}
    Let $X$ be an EFX allocation where all agents $j \in N_B$ have strictly more type $B$ items than all agents $i \in N_A$. 
    If there exists an agent $i \in N_A$ such that $|X_i| < |X_j|$ for all $j \in N_B$, then for all $i' \in N_A$ and $j' \in N_B$ it holds that agent $i'$ does not envy agent $j'$.
\end{lemma}
\begin{proof}  
    Assume that there exists an agent $i \in N_A$ such that $|X_i| < |X_j|$ for all $j \in N_B$.
    Consider any $i' \in N_A$.
    Since $X$ is EFX, we know that agent $i'$ does not EFX-envy agent $i$ and so $v_{i'}(X_{i'}) - v_{i'}^B \geq v_{i'}(X_{i})$. Rearranging, this tells us that $v_{i'}(X_{i'}) \geq v_{i'}(X_{i}) + v_{i'}^B$.
    Now, consider any $j' \in N_B$.
    Since $|X_i| < |X_{j'}|$ and $X_i$ has less type $B$ items than $X_{j'}$, it follows that $v_{i'}(X_i) + v_{i'}^B \geq v_{i'}(X_{j'})$.
    Therefore $v_{i'}(X_{i'}) \geq v_{i'}(X_i) + v_{i'}^B \geq v_{i'}(X_{j'})$ and so $i'$ does not envy $j'$.
\end{proof}

\begin{lemma}
    \label{lem:rule1 not applied twice}
    Let $X$ be an EFX partial allocation where all type $B$ items are allocated.
    Assume that \Cref{alg:efx} applies Rule 2 to $X$ to create a new allocation $X'$, and then applies Rule 1 to $X'$ to create $X''$.
    Then, the EFX condition of Rule 1 does not hold for $X''$.
\end{lemma}
\begin{proof}
    Let $X_i = (\alpha_i, \beta_i)$ for all $i \in N$. 
    Then, $X''_i = (\alpha_i, \beta_i)$ or $X''_i = (\alpha_i+1, \beta_i)$ for each $i \in N_A$ and $X''_j = (\alpha_j+1, \beta_j)$ for all $j \in N_B$. 
    
    Assume for contradiction that the EFX condition of Rule 1 holds for $X''$. 
    Let $Y''$ be the allocation formed if Rule 1 was applied to $X''$.
    Then, $Y''_i = X''_i$ for all $i \in N_A$ and $Y''_j = (\alpha_j+2, \beta_j)$ for all $j \in N_B$.
    Since the EFX condition of Rule 1 holds for $X''$, we know that $Y''$ is EFX.
   
    Let $Y$ be the allocation formed if Rule 1 was applied to $X$.
    Then, $Y_i = X_i$ for all $i \in N_A$ and $Y_j = (\alpha_j+1, \beta_j)$ for all $j \in N_B$.
    We will show that $Y$ is EFX.
    Firstly, no agent $i \in N_A$ has any EFX-envy because $X$ is EFX and $Y_i = X_i$ for all $i \in N_A$. 
    Now, consider an agent $j \in N_B$. 
    Agent $j$ does not EFX-envy any other agent $j' \in N_B$ because $X$ is EFX and both agents $j$ and $j'$ gained a single type $A$ item when comparing $Y$ to $X$.
    We now consider envy from $j$ towards an agent $i \in N_A$. 
    We consider two cases:
    \begin{itemize}
      \item First, assume that $\beta_j > 0$. Then, because $Y''$ is EFX we know that $v_j((\alpha_j+2, \beta_j-1)) \geq v_j(Y''_i) \geq v_j((\alpha_i+1, \beta_i))$ for all $i \in N_A$. Therefore, $v_j((\alpha_j+1, \beta_j-1)) \geq v_j((\alpha_i, \beta_i)) = v_j(Y_i)$ and so agent $j$ does not EFX-envy agent $i$ for all $i \in N_A$. 
      \item Now, assume that $\beta_j = 0$. Then, because $Y''$ is EFX we know that $v_j((\alpha_j+1, 0)) \geq v_j(Y''_i) \geq v_j((\alpha_i+1, \beta_i))$ for all $i \in N_A$. Therefore, $v_j((\alpha_j, 0)) \geq v_j((\alpha_i, \beta_i)) = v_j(Y_i)$ and so agent $j$ does not EFX-envy agent $i$ for all $i \in N_A$. 
    \end{itemize}
     Hence $Y$ is EFX.
     However, this is a contradiction because \Cref{alg:efx} would have applied Rule 1 to $X$ instead of Rule 2. 
\end{proof}

We are now ready to state our first set of sufficient conditions for the initial allocation $X^*$.

\begin{lemma}
    \label{using the update rules}
    Let $X^*$ be an EFX partial allocation where all type $B$ items are allocated.
    If $X^*$ satisfies the following conditions, then the update rules can be applied until all items are allocated:
    \begin{enumerate}
        \item The EFX condition of Rule 1 does not hold for $X^*$, and
        \item Consider a partial allocation $Y$ formed by applying the update rules 0 or more times to $X^*$.
        Whenever the EFX condition of Rule 1 does not hold for $Y$, Rule 2 can be applied $|N_B|$ times to $Y$.
    \end{enumerate}
\end{lemma}
\begin{proof}
    For brevity, within this proof we refer to the two conditions of the lemma as Condition 1 and Condition 2 respectively.
    We must show that Rule 2 can be applied whenever Rule 1 cannot be applied.
    
    Let $a_t$ be the number of unallocated type $A$ items after the update rules have been applied $t$ times to the allocation $X^*$.
    We begin by considering the case where $a_0 \leq |N_B|$. By Condition 1, we know that the EFX condition of Rule 1 does not hold for $X^*$. Hence, because of Condition 2 we can apply Rule 2 $a_0$ times to $X^*$, and so the lemma holds in this case.
    
    Otherwise, assume $a_0 > |N_B|$.
    We begin by showing that immediately after Rule 1 is applied, the EFX condition of Rule 1 will no longer hold.
    This follows from \Cref{lem:rule1 not applied twice}.
    In particular, from Condition 1 we know that Rule 2 will be the first rule applied to $X^*$.
    Thus, prior to every application of Rule 1 there must have been an application of Rule 2. 
    Hence, immediately after Rule 1 is applied, the EFX condition of Rule 1 will no longer hold.
     
    Now, we show that Rule 2 can be applied whenever Rule 1 cannot.
    Consider a situation where the update rules have been applied $t$ times, and Rule 1 cannot be applied.
    If $a_t \geq |N_B|$, then the EFX condition of Rule 1 must not hold and so Rule 2 can be applied because of Condition 2.
    If $a_t < |N_B|$, then consider the last update rule applied when $a_{t'} \geq |N_B|$ still held (that is, consider the largest $t'$ such that $a_{t'} \geq |N_B|$, and consider the next update rule applied).
    If it was Rule 1, then after this update the EFX condition of Rule 1 did not hold (as we showed in the previous paragraph) and so Rule 2 can be applied until every item is allocated. 
    If it was Rule 2, then the EFX condition of Rule 1 did not hold and so by Condition 2, Rule 2 can be applied until every item is allocated. 
\end{proof}

We now present \Cref{can apply rule 2 many times}, that gives a set of conditions under which the second condition of \Cref{using the update rules} is satisfied. 

\begin{lemma}
    \label{can apply rule 2 many times}
    Let $X$ be an EFX partial allocation.
    If $X$ satisfies the following conditions, then Rule 2 can be applied $|N_B|$ times:
    \begin{enumerate}
        \item For all agents $i \in N_A$ and $j \in N_B$, $X_j$ has strictly more type $B$ items than $X_i$,
        \item For all agents $i \in N_A$ and $j \in N_B$, $i$ does not envy $j$, and
        \item Consider a partial allocation $Y$ formed by applying the update rules 0 or more times to $X$.
        For any such allocation $Y$ and any nonempty subset $S \subseteq N_A$, there exists some agent $i \in S$ who does not envy any other agent in $S$.
    \end{enumerate}
\end{lemma}
\begin{proof}
    For brevity, within this proof we refer to the conditions of the lemma as Conditions 1-3 respectively.
    Let $Y^k$ be the allocation formed by applying rule 2 $k$ times to $X$. We will show that rule 2 can be applied to all $Y^k$ where $0 \le k < |N_A|$. 
    By induction, this means that we can apply rule 2 $|N_A| \geq |N_B|$ times to $X$.
    
    Assume that $k < |N_A|$ and that $Y^k$ is an EFX allocation.
    Let $T^k \subseteq N_A$ be the set of agents in $N_A$ who have the same bundle in $X$ and $Y^k$ (that is, $X_i = Y^k_i$).
    Since $Y^k$ is formed by applying rule 2 $k$ times to $X$, it follows that $|T^k| \geq |N_A| - k$ and so $|T^k| > 0$. 
    We will show that rule 2 can be applied to $Y^k$, which will create $Y^{k+1}$.
    
    From Condition 3, there exists an agent $t \in T^k$ such that $t$ does not envy any other agent in $T^k$ when the allocation is $Y^k$.
    Since $Y^k_t = X_t$, we know from Condition 2 that $t$ does not envy any $j \in N_B$.
    Hence, we can apply Rule 2 to agent $t$ if they do not envy any agent $i \in N_A \backslash T^k$.
    
    Otherwise, assume that $t$ envies some agent $i \in N_A \backslash T^k$. 
    Let $Y^k_i = (\alpha^k_i, \beta^k_i)$, $X_i = (\alpha_i, \beta_i)$ and $X_t = Y^k_t = (\alpha_t, \beta_t)$.
    Since $i \not\in T^k$, we know that $Y^k_i \neq X_i$ (in particular, $\alpha^k_i > \alpha_i$ and $\beta^k_i = \beta_i$).
    
    \noindent\textbf{Claim 1:} $Y^k_t$ contains only type $B$ items. That is, $\alpha_t = 0$.
    \begin{proof}[Proof of Claim 1] Assume for contradiction that $\alpha_t > 0$. Since the allocation $X$ is EFX, we know that $v_t((\alpha_t-1, \beta_t)) \geq v_t((\alpha_i, \beta_i))$ and so $v_t((\alpha_t, \beta_t)) \geq v_t((\alpha_i+1, \beta_i))$.
    Additionally, since $\alpha^k_i > \alpha_i$ and $\beta^k_i = \beta_i$, it follows that $v_t((\alpha_i+1, \beta_i)) \geq v_t((\alpha^k_i, \beta^k_i))$ and so $v_t(Y^k_t) = v_t((\alpha_t, \beta_t)) \geq v_t((\alpha^k_i, \beta^k_i)) = v_t(Y^k_i)$.
    However, by the choice of $i$ we know that $t$ envies $i$ when the allocation is $Y^k$, which implies that $v_t(Y^k_t) < v_t(Y^k_i)$, a contradiction.
\end{proof}    
    By Condition 3, there exists some agent $i' \in N_A$ who does not envy any other agent in $N_A$.
    From Claim 1 and Condition 1, we know that $Y^k_t \subset Y^k_j$ for all $j \in N_B$. 
    It follows that $i'$ must not envy any agent $j \in N_B$, since $j$ has a strictly worse bundle that $t$ and $i'$ does not envy $t$.
    Hence $i'$ does not envy any agent in $N$, and so Rule 2 can be applied to $i'$ to create the EFX allocation $Y^{k+1}$. 
\end{proof}
Finally, we provide a result which gives an alternate set of conditions for the initial allocation $X^*$.

\begin{lemma}
    \label{N_B same size}
    Let $X^*$ be an EFX partial allocation where all type $B$ items are allocated.
    If $X^*$ satisfies the following conditions, then the update rules can be applied until all items are allocated:
    \begin{enumerate}
        \item The EFX condition of Rule 1 does not hold for $X^*$,
        \item $|X^*_i| = |X^*_{i'}|$ for all $i, i' \in N_B$,
        \item For all agents $i \in N_A$ and $j \in N_B$, $X^*_j$ has strictly more type $B$ items than $X^*_i$, and
        \item Consider a partial allocation $Y$ formed by applying the update rules 0 or more times to $X^*$.
        For any such allocation $Y$ and any nonempty subset $S \subseteq N_A$, there exists some agent $i \in S$ who does not envy any other agent in $S$.
    \end{enumerate}
\end{lemma}
\begin{proof}
    For brevity, within this proof we refer to the conditions of the lemma as Conditions 1-4 respectively.
    We use \Cref{using the update rules}.
    The first condition of \Cref{using the update rules} is the same as Condition 1, and so we just need to show that the second condition of \Cref{using the update rules} is met. 
    
    Consider a partial allocation $Y$ formed by applying the update rules 0 or more times to $X^*$, and assume that the EFX condition of Rule 1 does not hold for $Y$. 
    We use \Cref{can apply rule 2 many times} to show that Rule 2 can be applied $|N_B|$ times to $Y$.
    The first and third conditions of \Cref{can apply rule 2 many times} immediately hold because they are shared with \Cref{N_B same size}.
    Hence, we just need to show that the second condition of \Cref{can apply rule 2 many times} holds.
    
    Let $i \in N_A$ and $j \in N_B$ be agents such that $j$ would EFX-envy $i$ if Rule 1 was applied to $Y$ (note that these agents must exist because the EFX condition of Rule 1 does not hold for $Y$). 
    By \Cref{envy means less items}, $|Y_i| < |Y_j| = |Y_{j'}|$ for all $j' \in N_B$.
    Hence, by \Cref{lem:less items implies no envy}, we know that for all agents $i' \in N_A$ and $j' \in N_B$, agent $i'$ does not envy agent $j'$.
    Thus \Cref{can apply rule 2 many times} holds for $Y$ and so \Cref{using the update rules} holds for $X^*$.
\end{proof}

\subsection{Computing \texorpdfstring{$X^*$}{}}\label{subsection: deciding the initial allocation}

In this section, we describe how to compute $X^*$ and justify how this initial allocation is sufficient for \Cref{alg:efx} to output an EFX allocation.
We consider several cases, depending on the input instance.

Let $a$ and $b$ be the number of unallocated type $A$ and $B$ items respectively. 
Initially, $a = |A|$ and $b = |B|$.

Let $k = \left\lfloor \frac{b-|N_B|}{n} \right\rfloor$.
We begin by assigning $k$ type $B$ items to all agents in $N_A$ and $k+1$ to all agents in $N_B$.
In particular,
\begin{equation*}
    X^*_i =
    \begin{cases}
      (0, k) & \text{for $i \in N_A$,} \\
      (0, k+1) & \text{for $i \in N_B$.} \\
    \end{cases}
\end{equation*}

Now, $0 \leq b < n$.
We consider two cases, depending on $b$.

\subsubsection{Case 1: \texorpdfstring{$b \geq |N_B|$}{}}
\label{case 1: b leq |N_B|}
Let $N_A' = \{ \,  i \in N_A : i > n-b \,  \}$.
Note that $|N_A'| = b - |N_B|$.
We allocate one more type $B$ item to all agents in $N_B \cup N_A'$, so that $b = 0$. We also allocate one type $A$ item to all agents in $N_A \backslash N_A'$.
In particular, the partial allocation is:

\begin{equation*}
    X^*_i =
    \begin{cases}
      (1, k) & \text{for $i \in N_A \backslash N_A'$,} \\
      (0, k+1) & \text{for $i \in N_A'$,} \\
      (0, k+2) & \text{for $i \in N_B$.} \\
    \end{cases}
\end{equation*}

We use \Cref{N_B same size} to show that the update rules can be applied until all items are allocated. 
First, note that the partial allocation is EFX and the first three conditions of \Cref{N_B same size} clearly hold. 
For the fourth condition, consider a partial allocation $Y$ as described in \Cref{N_B same size}, and some nonempty subset $S \subseteq N_A$.
If $S \subseteq N_A'$ or $S \subseteq N_A \setminus N_A'$, then the fourth condition holds as any agent $i \in \argmin_{j \in S}|Y_j|$ does not envy any other agents in $S$.
Otherwise, let $i$ be an agent in $\argmin_{j \in S \cap N_A'}|Y_j|$ and $i'$ be an agent in $\argmin_{j \in S \cap (N_A \setminus N_A')}|Y_j|$.
By \Cref{no mutual envy} these agents cannot both envy each other, and so assume without loss of generality that $i$ does not envy $i'$.
Then, $i$ does not envy any agents in $S$.
Hence this allocation satisfies all the conditions of \Cref{N_B same size}.

\subsubsection{Case 2: \texorpdfstring{$b < |N_B|$}{}}
Let $N_B' = \{ \, i \in N_B : i > n-b \, \}$.
Note that $|N_B'| = b$.
We assign one more type $B$ item to all agents in $N_B'$, so that $b = 0$. 
This gives us the following partial allocation that is not EFX:

\begin{equation*}
    X^*_i =
    \begin{cases}
      (0, k) & \text{for $i \in N_A$,} \\
      (0, k+1) & \text{for $i \in N_B \backslash N_B'$,} \\
      (0, k+2) & \text{for $i \in N_B'$.} \\
    \end{cases}
\end{equation*}

If $|A| \leq 2|N_A|$, then we allocate the type $A$ items to agents in $N_A$ in a round-robin way. 
Note that each agent in $N_A$ will receive 1 or 2 type $A$ items (since $|N_A| < |A| \leq 2|N_A|$).
In particular, let $N_A'$ be the agents who receive 1 type $A$ item.
Then we will have the following EFX allocation:
\begin{equation*}
    X^*_i =
    \begin{cases}
      (1, k) & \text{for $i \in N_A'$,} \\
      (2, k) & \text{for $i \in N_A \backslash N_A'$,} \\
      (0, k+1) & \text{for $i \in N_B \backslash N_B'$,} \\
      (0, k+2) & \text{for $i \in N_B'$.} \\
    \end{cases}
\end{equation*}
Since there are no unallocated items, this case is complete.
Otherwise, we know that $|A| > 2|N_A|$. We consider three final subcases.

\paragraph{Case 2.1: For all $j \in N_B \backslash N_B'$, agent $j$ does not strongly prefer $B$ (recall that $j$ strongly prefers $B$ if $2v_j^B \geq v_j^A$)}

In this case, we allocate one type $A$ item to all agents in $N_A \cup N_B \backslash N_B'$, resulting in the following EFX partial allocation:

\begin{equation*}
    X^*_i =
    \begin{cases}
      (1, k) & \text{for $i \in N_A$,} \\
      (1, k+1) & \text{for $i \in N_B \backslash N_B'$,} \\
      (0, k+2) & \text{for $i \in N_B'$.} \\
    \end{cases}
\end{equation*}

This partial allocation is EFX because agents in $N_B \backslash N_B'$ prefer 1 type $A$ item over 2 type $B$ items. 
We use \Cref{N_B same size} to show that the update rules can be applied until all items are allocated. 
The first three conditions of \Cref{N_B same size} clearly hold. 
For the fourth condition, consider a partial allocation $Y$ as described in \Cref{N_B same size} and a nonempty subset $S \subseteq N_A$.
Then, any agent $i \in \argmin_{j \in S}|Y_j|$ does not envy any other agents in $S$.
Hence this allocation satisfies all the conditions of \Cref{N_B same size}.

\paragraph{Case 2.2: There are at least $|N_B|$ agents $i \in N_A$ who strongly prefer $A$}

In this case, we give one type $A$ item to all agents in $N_A$, resulting in the following EFX partial allocation:

\begin{equation*}
    X^*_i =
    \begin{cases}
      (1, k) & \text{for $i \in N_A$,} \\
      (0, k+1) & \text{for $i \in N_B \backslash N_B'$,} \\
      (0, k+2) & \text{for $i \in N_B'$.} \\
    \end{cases}
\end{equation*}

We use \Cref{using the update rules} to show that the update rules can always be applied.
The first condition clearly holds.
For the second condition, consider a partial allocation $Y$ as described in \Cref{using the update rules}, and assume that the EFX condition of Rule 1 does not hold for $Y$. 
Then, there must exist some agent $j \in N_B$ who would EFX-envy some agent $i \in N_A$ if Rule 1 were to be applied. 
Then, by \Cref{envy means less items}, we know that $|Y_i| < |Y_{j}|$.
However, observe that $|Y_j| \leq |Y_{j'}|+1$ for all $j' \in N_B$ and so $|Y_i| < |Y_{j}| \leq |Y_{j'}|+1$ implying that $|Y_i| \leq |Y_{j'}|$ for all $j' \in N_B$.
Additionally, since all agents in $N_A$ have the same number of type $B$ items and $Y$ is EFX, it follows that $|Y_{i'}| \leq |Y_{i}|+1 \leq |Y_{j'}|+1$ for all $i' \in N_A$ and $j' \in N_B$.

We can therefore apply Rule 2 at least $|N_B|$ times to $Y$ as follows:
\begin{itemize}
    \item While there exists an agent $i \in N_A$ where $|Y_i| \leq |Y_j|$ for all $j \in N_B$, apply Rule 2 to such an agent with the smallest $|Y_i|$. This maintains EFX as $i$ did not envy any agent prior to the rule being applied.
    \item After doing the above step one or more times, all agents $i \in N_A$ have identical bundles (with $|Y_{i}| \leq |Y_j|+1$ for all $j \in N_B$). We can apply Rule 2 once to all agents who strongly prefer $A$. This maintains EFX as these agents will not EFX-envy any $j \in N_B$ because they prefer two type $A$ items over a type $B$ item.
\end{itemize}

\paragraph{Case 2.3: Cases 2.1 and 2.2 do not hold}
Since Case 2.2 does not hold, there are less than $|N_B|$ agents $i \in N_A$ who strongly prefer $A$.
Since Case 2.1 does not hold, there exists some agent $j \in N_B \backslash N_B'$ who strongly prefers $B$ and so all agents $j' \in N_B'$ must strongly prefer $B$.

Let $N_A' = \{ \, i \in N_A : i \leq |N_B| \,  \}$.
Note that $|N_A'| = |N_B|$.
We transfer one type $B$ item from each agent in $N_A'$ to the agents in $N_B$, allocate 2 type $A$ items to all agents in $N_A'$ and allocate 1 type $A$ item to all agents in $N_A \backslash N_A'$.
In particular,

\begin{equation*}
    X^*_i =
    \begin{cases}
      (2, k-1) & \text{for $i \in N_A'$,} \\
      (1, k) & \text{for $i \in N_A \backslash N_A'$,} \\
      (0, k+2) & \text{for $i \in N_B \backslash N_B'$,} \\
      (0, k+3) & \text{for $i \in N_B'$.} \\
    \end{cases}
\end{equation*}
Since there are less than $|N_B|$ agents $i \in N_A$ who strongly prefer $A$, all these agents must be in $N_A'$ and so no agent in $N_A \backslash N_A'$ strongly prefers $A$.
Thus, there is no EFX-envy from any agent $i \in N_A$ towards any other agent in $N$. 
Additionally, since all agents in $N_B'$ strongly prefer $B$ there is no EFX-envy from any agent $j \in N_B$ towards any other agent in $N$.
Therefore, $X^*$ is EFX.

We use \Cref{using the update rules} to show that this initial allocation is sufficient for \Cref{alg:efx}.
The first condition of  \Cref{using the update rules} clearly holds.
For the second condition, consider a partial allocation $Y$ as described in \Cref{using the update rules}, and assume that the EFX condition of Rule 1 does not hold for $Y$. 
We use \Cref{can apply rule 2 many times} to show that the second condition of \Cref{using the update rules} holds.
\begin{enumerate}
    \item The first condition of \Cref{can apply rule 2 many times} holds for $Y$ because it holds for $X^*$.
    \item For the second condition of \Cref{can apply rule 2 many times}, note that the EFX condition of $Y$ does not hold by the definition of $Y$. Hence, there exists some $i \in N_A$ and $j \in N_B$ such that $j$ would EFX-envy $i$ if Rule 1 was applied. 
    By \Cref{envy means less items}, $|Y_i| < |Y_j|$.

    If $j \in N_B \backslash N_B'$, then $|Y_i| < |Y_{j'}|$ for all $j' \in N_B$. Therefore, by \Cref{lem:less items implies no envy} we know that for all agents $i' \in N_A$ and $j' \in N_B$, agent $i'$ does not envy agent $j'$.
    
    If $j \in N_B'$, then we show that $|Y_i| < |Y_j|-1$, by proving that $|Y_i| \neq |Y_j|-1$.
    Let $Y_j = (\alpha_j, k+3)$ and assume $|Y_i| = |Y_j|-1$. 
    Then,
    \begin{equation*}
        Y_i =
        \begin{cases}
          (\alpha_j+3, k-1) & \text{if $i \in N_A'$,} \\
          (\alpha_j+2, k) & \text{if $i \in N_A \backslash N_A'$.} \\
        \end{cases}
    \end{equation*}
    If Rule 1 was applied, $Y_j$ would be $(\alpha_j+1, k+3)$.
    However, if this occurred, $j$ would not EFX-envy $i$ in either case (since $j$ strongly prefers $B$) and so $|Y_i| < |Y_j|-1$.
    This implies that $|Y_i| < |Y_{j'}|$ for all $j' \in N_B$ and so we can apply \Cref{lem:less items implies no envy}.    
    \item For the third condition of \Cref{can apply rule 2 many times}, we can use the same argument that is used in \Cref{case 1: b leq |N_B|}.
\end{enumerate}

This completes our proof of \Cref{thm: efx}.

\section{Algorithm for Checking Existence of EF Allocations}

For negative additive valuations, checking whether an envy-free allocation exists is NP-complete \citep{BSV21a}. Under our scenario of two chore types, we propose a polynomial-time algorithm to solve the problem. In particular, we prove the following result.

\begin{theorem}
  \label{thm: finding EF in polynomial time}
  For two chore type instances, an envy-free allocation can be found in polynomial-time (with respect to the number of agents and items) whenever one exists.
\end{theorem}

Before proving this theorem in full, let us first deal with a trivial case. If $v_i^A = 0$ and $v_j^B = 0$ for some (not necessarily distinct) agents $i$ and $j$, then we can allocate all chores of types $A$ and $B$ to agents $i$ and $j$ respectively. Since all other agents are not given any chores, the resulting allocation is trivially envy-free. It suffices therefore to only consider cases where at most one chore type is valued at zero by at least one agent.

To further simplify the problem, we also do the following: if $v_i^B = 0$ for some agent $i$, swap the chore types---that is, rename them---so that $v_i^A = 0$ instead. Then, without loss of generality, we may assume $v_i^A \leq 0$ and $v_i^B < 0$. To prove \Cref{thm: finding EF in polynomial time}, we first present a result about the structure of any envy-free allocation.

\begin{lemma}
    Consider a two chore types instance where $v_i^A \leq 0$ and $v_i^B < 0$ for all $i \in N$.
    Let $X$ be an envy-free allocation and $i$ and $j$ be two agents with bundles $X_i = (\alpha_i, \beta_i)$ and $X_j = (\alpha_j, \beta_j)$. If $\frac{v_i^A}{v_i^B} < \frac{v_j^A}{v_j^B}$, then $\alpha_i \geq \alpha_j$ must hold.
\end{lemma}
\begin{proof}
    Assume by contradiction that $\frac{v_i^A}{v_i^B} < \frac{v_j^A}{v_j^B}$ but $\alpha_i < \alpha_j$. Agent $i$ does not envy agent $j$, so $v_i(X_i) \geq v_i(X_j)$, i.e.,
      \[ \alpha_i v_i^A + \beta_i v_i^B \geq \alpha_j v_i^A + \beta_j v_i^B . \]
    Rearranging,
      \[
        (\alpha_i - \alpha_j) v_i^A \geq (\beta_j - \beta_i) v_i^B .
      \]
    Similarly, since agent $j$ does not envy agent $i$, we have
      \[
        (\alpha_j - \alpha_i) v_j^A \geq (\beta_i - \beta_j) v_j^B .
      \]
    Since $v_i^B$, $v_j^B$, and $\alpha_i - \alpha_j$ are strictly negative, we have
      \[
        \frac{v_i^A}{v_i^B} \geq \frac{\beta_j - \beta_i}{\alpha_i - \alpha_j}
          \qquad \text{and} \qquad
        \frac{v_j^A}{v_j^B} \leq \frac{\beta_i - \beta_j}{\alpha_j - \alpha_i}
          = \frac{\beta_j - \beta_i}{\alpha_i - \alpha_j} .
      \]
    By assumption,
      \[
        \frac{\beta_j - \beta_i}{\alpha_i - \alpha_j}
          \leq \frac{v_i^A}{v_i^B}
          < \frac{v_j^A}{v_j^B}
          \leq \frac{\beta_j - \beta_i}{\alpha_i - \alpha_j} ,
      \]
    which is absurd.
\end{proof}

\begin{corollary}
    \label{cor: EF sorted order}
    Consider a two chore types instance where $v_i^A \leq 0$ and $v_i^B < 0$ for all $i \in N$.
    Let $X$ be an envy-free allocation where each agent $i$ is allocated the bundle $X_i = (\alpha_i, \beta_i)$. We can reorder the agents so that $\frac{v_j^A}{v_j^B} \leq \frac{v_{j+1}^A}{v_{j+1}^B}$ and $\alpha_j \geq \alpha_{j + 1}$ for all $1 \leq j < n$.
\end{corollary}

For the remainder of this section, we assume the agents are reordered as in \Cref{cor: EF sorted order}. The following result provides an easy method for checking envy-freeness.

\begin{lemma}
    \label{lem: local EF implies EF}
        Consider a two chore types instance where $v_i^A \leq 0$ and $v_i^B < 0$ for all $i \in N$, and $\frac{v_i^A}{v_i^B} \leq \frac{v_{i+1}^A}{v_{i+1}^B}$ for all $1 \leq i < n$.
        Let $X$ be an allocation where each agent $i$ receives the bundle $X_i = (\alpha_i, \beta_i)$, and $\alpha_i \geq \alpha_{i + 1}$ for all $1 \leq i < n$. If agents $i$ and $i + 1$ do not envy each other for all $1 \leq i < n$, then the allocation $X$ is envy-free.
\end{lemma}
\begin{proof}
    Let $i, j, k$ be three agents where $i < j < k$. It is sufficient to prove that non-envy between these agents is transitive: that is, whenever agents $i$ and $j$ do not envy each other, and agents $j$ and $k$ do not envy each other, then $i$ and $k$ do not envy each other. Let us show this is the case.

    By assumption, $\alpha_i \geq \alpha_j \geq \alpha_k$. Since agents $i$ and $j$ do not envy each other, this implies $\beta_i \leq \beta_j$. Similarly, since agents $j$ and $k$ do not envy each other, $\beta_j \leq \beta_k$.

    Let us first consider the case where at least one of $i$, $j$ and $k$ is indifferent towards type $A$ chores. By the ordering assumption, this forces $v_i^A = 0$. Then, $v_i(X_i) = v_i^B \beta_i \geq v_i^B \beta_k = v_i(X_k)$. Hence agent $i$ never envies agent $k$ in this case. It remains to show agent $k$ does not envy agent $i$: to do this, let us consider three subcases.
      \begin{enumerate}
        \item If $v_i^A = v_j^A = v_k^A = 0$, then $\beta_i = \beta_j = \beta_k$. Then $v_k(X_k) = v_k^B \beta_k = v_k^B \beta_i = v_k^B(X_i)$, so agent $k$ does not envy agent $i$.
        \item Suppose $v_i^A = v_j^A = 0$, but $v_k^A < 0$. This implies $\beta_i = \beta_j$. Then $v_k(X_k) \geq v_k(X_j) = v_k((\alpha_j, \beta_j)) \geq v_k((\alpha_i, \beta_i)) = v_k(X_i)$. Hence agent $k$ does not envy agent $i$.
        \item The last subcase, where $v_i^A = 0$ but $v_j^A$ and $v_k^A$ are strictly negative, is deferred: this is handled by the method detailed below.
      \end{enumerate}
    Therefore, when at least one of agents $i, j, k$ is indifferent towards type $A$ chores, agents $i$ and $k$ do not envy each other. To complete the proof, let us now consider the remaining case where $v_i^A, v_j^A, v_k^A$ are strictly negative.

    Since agent $j$ does not envy agent $k$, we have
      \[
        \alpha_j v_j^A + \beta_j v_j^B \geq \alpha_k v_j^A + \beta_k v_j^B .
      \]
    Rearranging,
      \begin{equation} \label{eq:ef-ineq}
        (\alpha_j - \alpha_k) v_j^A \geq (\beta_k - \beta_j) v_j^B .
      \end{equation}
    Since $\beta_j \leq \beta_k$, both sides of the inequality are non-positive. Also, since $\frac{v_i^A}{v_i^B} \leq \frac{v_j^A}{v_j^B}$, we have $\frac{v_i^A}{v_j^A} \leq \frac{v_j^A}{v_j^B}$. Multiplying the left side of inequality (\ref{eq:ef-ineq}) by $\frac{v_i^A}{v_j^A}$ and the right side by $\frac{v_i^B}{v_j^B}$ yields
      \[
        (\alpha_j - \alpha_k) v_i^A \geq (\beta_k - \beta_j) v_i^B ,
      \]
    so $v_i(X_j) \geq v_i(X_k)$. By assumption, agent $i$ does not envy agent $j$, i.e., $v_i(X_i) \geq v_i(X_j)$; hence agent $i$ does not envy agent $k$ either.

    We now use a similar approach to show that agent $k$ does not envy agent $i$; this method also deals with case 3 from earlier, since it applies even if $v_i^A = 0$. Because agent $j$ does not envy agent $i$, we have
      \[
        (\alpha_j - \alpha_i) v_j^A \geq (\beta_i - \beta_j) v_j^B . 
      \]
    In this case, both sides of the inequality are non-negative. Since $\frac{v_k^A}{v_j^A} \geq \frac{v_k^B}{v_j^B}$, it follows that
      \[
        (\alpha_j - \alpha_i) v_k^A \geq (\beta_i - \beta_j) v_k^B ,  
      \]
    so $v_k(X_j) \geq v_k(X_i)$. By assumption, $v_k(X_k) \geq v_k(X_j)$, so agent $k$ does not envy agent $i$ either.
\end{proof}

\begin{algorithm}
\caption{Dynamic programming function to determine whether a partial allocation can be extended into an envy-free allocation.}
\label{alg:dpforef}
\SetKwFunction{Ff}{f}
\SetKwProg{Fn}{Function}{:}{}
\SetKwComment{Comment}{$\triangleright$\ }{}

\Fn(){\Ff{$a$, $b$, $i$, $\alpha$, $\beta$}} { 
\If(\Comment*[f]{Base case: there are no agents remaining}){$i = n$} {
	\If{$a+b = 0$} {
		\Return YES
	}
	\Else { 
		\Return NO
	}
}
\Comment*[f]{Try every possible bundle $(\alpha', \beta')$ for agent $i+1$ such that $\alpha' \leq \alpha$}

\For{$\alpha' \gets 0$ \KwTo $\min(a, \alpha)$}{
	\For{$\beta' \gets 0$ \KwTo $b$}{
		\If{If $(\alpha', \beta')$ can be assigned to agent $i+1$ such that agent $i$ and $i+1$ are envy-free} {
			\If{$f(a-\alpha', b-\beta', i+1, \alpha', \beta')$ is YES} {
			\Return YES
			}
		}
	}
}
\Return NO
}
\end{algorithm}

We use \Cref{lem: local EF implies EF} to create a dynamic programming algorithm, \Cref{alg:dpforef}, to help us determine whether an envy-free allocation exists.
Let $f(a, b, i, \alpha, \beta)$ be the result of a subproblem that represents a state where we have assigned bundles to the first $i$ agents.
In particular, the state $(a, b, i, \alpha, \beta)$ represents the following:
\begin{itemize}
	\item Items have been assigned to the first $i$ agents such that they are envy-free,
	\item Agent $i$ received the bundle $(\alpha, \beta)$, and
	\item There are $a$ type $A$ items and $b$ type $B$ items to allocate to the remaining $n-i$ agents.
\end{itemize}

The result of $f(a, b, i, \alpha, \beta)$ is YES if the remaining items can be allocated in an envy-free way to agents $i+1$ through $n$, and NO otherwise.

To compute $f(a, b, i, \alpha, \beta)$, every valid bundle for agent $i+1$ is considered.
In particular, we consider every bundle $(\alpha', \beta')$ satisfying $\alpha' \leq a$, $\beta' \leq b$ and $\alpha' \leq \alpha$. 
If there exists such a bundle $(\alpha', \beta')$ that can be extended into an envy-free allocation, then the result is YES. Otherwise, the result is NO.
Envy-freeness is checked using \Cref{lem: local EF implies EF}.
Correctness of \Cref{alg:dpforef} holds because it considers every assignment satisfying the structure of \Cref{cor: EF sorted order}.

Since there are polynomial many states and each state takes polynomial-time to compute (with respect to the number of agents and items), the dynamic programming algorithm runs in polynomial time.

We use \Cref{alg:dpforef} to find an envy-free allocation whenever one exists.
In particular, for each agent starting from agent 1, we try every possible bundle until one is found that can be extended into an envy-free allocation. 
If this procedure succeeds, then we have found an envy-free allocation in polynomial time.
If this procedure fails, then by \Cref{cor: EF sorted order} and \Cref{lem: local EF implies EF} we know that there does not exist any envy-free allocation.

\begin{remark}
    \label{remark: finding EF in polynomial time for goods}
    We note that the same approach can be used to prove the equivalent result of \Cref{thm: finding EF in polynomial time} for two \emph{good} types.
\end{remark}

\section{Discussion}

The existence of EF1 and PO allocations or EFX allocations for the case of chores are major open problems in fair division.  
In this paper, we identified a natural setting or valuation restriction under which not only can we guarantee the existence of allocations that satisfy EF1 and PO, and EFX respectively, but such allocations can be computed in polynomial time. 
A related question is the complexity of checking whether there exists an envy-free allocation. Whereas this problem is NP-complete for chores in general,  we showed that there exists a dynamic program  for two chore types instances that can solve the problem in polynomial time.
There are several relevant problems that remain open. The existence and complexity of EF1 and PO allocations or EFX allocations is open for personalized bi-valued utilities. It is also open whether there always exists a PO and EFX allocation for our setting. 

\section*{Acknowledgment}
Aziz is supported by the Defence Science and Technology Group through the Centre for Advanced Defence Research in Robotics and Autonomous Systems under the project “Task Allocation for Multi-Vehicle Coordination” (UA227119). Mashbat Suzuki is supported by the ARC Laureate Project FL200100204 on "Trustworthy AI".

\bibliographystyle{ACM-Reference-Format}  

\bibliography{abb.bib,haris_master_angus.bib, aziz_personal_angus.bib}

\clearpage

\appendix

\section{EFX: Failure of Existing Approaches} 
\label{section: Failure of Existing EFX Approaches}

In this section, we explore important algorithms for chore allocation as well natural adaptations for fair allocation of goods to the case of chores. 
Our finding is that the algorithms do not give the EFX guarantee even for the case of two item types, which suggests that a different approach is required to find an EFX allocation.

\subsection{Goods algorithm of \citet{GMV22a}}

In  the paper~\citep{GMV22a}, an EFX algorithm is presented for goods with two item types. 
In their algorithm, they begin by allocating each agent their most preferred item in a round-robin way.
This process stops once there are not enough items remaining to continue this. 
They then describe how to allocate the remaining items. 

We show that there exists a case with chores where this approach cannot produce an EFX allocation.
The case has 4 agents, numbered from 1 to 4, and 6 items (3 of type $A$ and 3 of type $B$).
The valuations are as follows, where $\epsilon$ is a sufficiently small positive constant. Note that $N_A = \{ 1 \}$ and $N_B = \{ 2, 3, 4 \}$.

\begin{table}[h]
    \centering
    \begin{tabular}{ccc}
    	\toprule
	\textbf{Agents} & \textbf{Valuation of type $A$ items} & \textbf{Valuation of type $B$ items} \\
	\midrule
	1 & $-\frac{1}{6}+\epsilon$ & $-\frac{1}{6}-\epsilon$ \\
 	2, 3, 4 & $-\frac{1}{6}-\epsilon$ & $-\frac{1}{6}+\epsilon$ \\
	\bottomrule
    \end{tabular}
    \caption{\label{table:goods-algo-table} Two chore type instance where the goods algorithm of \citet{GMV22a} fails to find an EFX allocation.}
\end{table}

If we allocate each agent their most preferred item in a round robin way, this creates a partial allocation where $X_1 = (1, 0)$ and $X_2 = X_3 = X_4 = (0, 1)$, with 2 unallocated type $A$ items. 
We cannot allocate both of these to one agent, as this would not be EFX.
Hence, at least one agent in $N_B$ must receive one of the unallocated chores, and at least one agent from $N_B$ must not receive one of the unallocated chores.
However, this is not EFX as an agent in $N_B$ with $X_i = (1, 1)$ would EFX-envy an agent with $X_j = (0, 1)$.

\subsection{PROPX Algorithms of \citet{LLW21a}}
\citet{LLW21a} provide two algorithms which produce a PROPX allocation.
They begin by transforming any instance into an instance with identical ordering.
An instance has identical ordering (IDO) if all agents agree on the ordering of the items.
In particular, let $c_1, c_2, ..., c_m$ be the chores.
Then, $v_i(c_1) \leq v_i(c_2) \leq ... \leq v_i(c_m)$ for all agents $i$.
They then use one of two algorithms, ``The Top-trading Envy Cycle Elimination Algorithm'' and ``The Bid-and-Take Algorithm'', to create a PROPX allocation for the IDO instance. 
They then provide a mechanism to transform this into a PROPX allocation for the original non-IDO instance. 
We show that there exists a case with two item types where both algorithms create an allocation that is not EFX.

The case has 3 agents, numbered from 1 to 3, and 6 items (3 of type $A$ and 3 of type $B$).

The valuations are in \Cref{table:propx-table-1}, where $\epsilon$ is a sufficiently small positive constant. Note that $N_A = \{ 1 \}$ and $N_B = \{ 2, 3 \}$.

\begin{table}[h]
    \centering
    \begin{tabular}{ccc}
    	\toprule
	\textbf{Agent} & \textbf{Valuation of type $A$ items} & \textbf{Valuation of type $B$ items} \\
	\midrule
	1 & $-3\epsilon$ & $-\frac{1}{3}+3\epsilon$ \\
	2 & $-\frac{1}{3}+2\epsilon$ & $-2\epsilon$ \\
	3 & $-\frac{1}{3}+\epsilon$ & $-\epsilon$ \\
	\bottomrule
    \end{tabular}
    \caption{\label{table:propx-table-1} Two chore type instance where PROPX algorithms fail to find an EFX allocation.}
\end{table}

This instance is transformed into an instance with identical ordering, as shown in \Cref{table:propx-table-2}.

\begin{table}[h]
    \centering
    \begin{tabular}{ccc}
    	\toprule
	\textbf{Agent} & \textbf{Valuation of type $A$ items} & \textbf{Valuation of type $B$ items} \\
	\midrule
	1 & $-\frac{1}{3}+3\epsilon$ & $-3\epsilon$ \\
 	2 & $-\frac{1}{3}+2\epsilon$ & $-2\epsilon$ \\
	3 & $-\frac{1}{3}+\epsilon$ & $-\epsilon$ \\

	\bottomrule
    \end{tabular}
    \caption{\label{table:propx-table-2} The instance from \Cref{table:propx-table-1}, transformed into an instance with identical ordering.}
\end{table}

\paragraph{The Top-trading Envy Cycle Elimination Algorithm:} In this algorithm, items are allocated from the least valuation to the greatest valuation (according to the IDO instance) 
to an agent who does not envy any other agent. 
This leads to an allocation where each agent receives one type $A$ and one type $B$ item.

This allocation is then transformed into a PROPX allocation for the non-IDO case. 
This leads to one of the following two allocations, depending on the tiebreaking used:
\begin{itemize}
    \item $X_1 = (2, 0) $, $X_2 = (1, 1)$ and $X_3 = (0, 2)$.
    \item $X_1 = (2, 0) $, $X_2 = (0, 2)$ and $X_3 = (1, 1)$.
\end{itemize}
Neither allocation is EFX.
In the first case, this is due to the envy that agent 2 has for agent 3, and in the second case this is due to the envy that agent 3 has for agent 2.

\paragraph{The Bid-and-Take Algorithm:} In this algorithm, items are allocated from the least valuation to the greatest valuation (according to the IDO instance) to an agent which has the greatest valuation for this item, as long as this satisfies PROPX.
This leads to the following allocation:
\begin{itemize}
    \item $X_1 = (2, 0) $, $X_2 = (1, 1)$ and $X_3 = (0, 2)$.
\end{itemize}

This allocation is then transformed into a PROPX allocation for the non-IDO case, which leaves the allocation unchanged.
This is not EFX due to the envy that agent 2 has for agent 3.

\end{document}